\documentclass[11pts]{article}

\usepackage{enumerate}
\usepackage{hyperref}
\usepackage[utf8]{inputenc}
\usepackage[T1]{fontenc}
\usepackage[english]{babel}
\usepackage{url}
\usepackage[squaren, Gray, cdot]{SIunits}
\usepackage{array}
\usepackage{caption}
\usepackage{layout}
\usepackage{csquotes}
\usepackage[top=3.2cm,bottom=3.2cm,left=3.2cm,right=3.2cm]{geometry}

\usepackage{amsthm}
\usepackage{amsmath}
\usepackage{amssymb}
\usepackage{thm-restate}

\newtheorem{theorem}{Theorem}
\newtheorem{definition}[theorem]{Definition}
\newtheorem{lemma}[theorem]{Lemma}
\newtheorem{proposition}[theorem]{Proposition}
\newtheorem{claim}[theorem]{Claim}
\newtheorem{remark}[theorem]{Remark}
\newtheorem{corollary}[theorem]{Corollary}

\usepackage{listings}

\usepackage{graphicx}
\usepackage{subfigure}

\usepackage{bbm}
\usepackage{stmaryrd}
\usepackage{tikz}

\usepackage{algorithm}
\usepackage[noend]{algpseudocode}

\newcommand{\rnk}{\mathrm{rank}}
\newcommand{\spn}{\mathrm{span}}

\allowdisplaybreaks

\title{Parameterized Matroid-Constrained Maximum Coverage}

\date{}

\author{François Sellier\\
\textit{Université Paris Cité, CNRS, IRIF, Paris, France}\\
\textit{Mines Paris, Université PSL, Paris, France}}

\begin{document}

\maketitle

\begin{abstract}
    In this paper, we introduce the concept of \emph{Density-Balanced Subset} in a matroid, in which independent sets can be sampled so as to guarantee that (i) each element has the same probability to be sampled, and (ii) those events are negatively correlated. These \emph{Density-Balanced Subsets} are subsets in the ground set of a matroid in which the traditional notion of uniform random sampling can be extended.
    
    We then provide an application of this concept to the \emph{Matroid-Constrained Maximum Coverage} problem. In this problem, given a matroid $\mathcal{M} = (V, \mathcal{I})$ of rank $k$ on a ground set $V$ and a coverage function $f$ on $V$, the goal is to find an independent set $S \in \mathcal{I}$ maximizing $f(S)$. This problem is an important special case of the much-studied submodular function maximization problem subject to a matroid constraint; this is also a generalization of the maximum $k$-cover problem in a graph. In this paper, assuming that the coverage function has a bounded frequency $\mu$ (\emph{i.e.}, any element of the underlying universe of the coverage function appears in at most $\mu$ sets), we design a procedure, parameterized by some integer $\rho$, to extract in polynomial time an approximate kernel of size $\rho \cdot k$ that is guaranteed to contain a $1 - (\mu - 1)/\rho$ approximation of the optimal solution. This procedure can then be used to get a Fixed-Parameter Tractable Approximation Scheme (FPT-AS) providing a $1 - \varepsilon$ approximation in time $(\mu/\varepsilon)^{O(k)} \cdot |V|^{O(1)}$. This generalizes and improves the results of~[Manurangsi, 2019] and~[Huang and Sellier, 2022], providing the first FPT-AS working on an arbitrary matroid. Moreover, because of its simplicity, the kernel construction can be performed in the streaming setting.
\end{abstract}


\paragraph*{Funding} This work was funded by the grant ANR-19-CE48-0016 from the French National Research Agency (ANR)

\paragraph*{Acknowledgements} The author thanks Chien-Chung Huang, Claire Mathieu, Eli Upfal, and the anonymous reviewers for their helpful comments.

\section{Introduction}

    \paragraph*{Overview} Matroids are fundamental combinatorial structures that generalize the notion of linear independence in a vector space as well as the notion of forest in a graph. In combinatorial optimization, the matroid constraints are an important generalization of the cardinality constraint. For instance, consider the problem of maximizing a submodular function under some constraint. If the constraint is that the feasible subsets are those of size bounded by some parameter $k$ (cardinality constraint), an approximation of $1 - 1/e$ can be obtained in polynomial time by a simple greedy algorithm~\cite{nemhauser1978analysis} (this ratio is also the best possible in polynomial time unless $P = NP$, see~\cite{Feige98}). Under the more general constraint that the feasible subsets are those that are independent in a given matroid $\mathcal{M}$ (matroid constraint), an approximation of $1-1/e$ can also be achieved in polynomial time~\cite{CalinescuCPV11}, albeit by using a more involved continuous greedy technique. 
    Hence, somehow surprisingly, even though the matroid constraint is more complex than the cardinality constraint, some optimization problems are not really ``harder'' in the matroid context. 
    
    Following this perspective, in this paper we consider the problem of maximizing a coverage function of bounded frequency under some constraint. The starting point of our paper is the work of Manurangsi~\cite{Manurangsi19-sosa} in which, for cardinality constraints, a Fixed-Parameter Tractable Approximation Scheme (FPT-AS) is developed. Our main result here is a generalization of that FPT-AS to matroid constraints (Corollary~\ref{cor:fpt-as}), extending the approximate kernel construction of~\cite{Manurangsi19-sosa} to matroids (Theorem~\ref{thm:ratio-union}). A key idea in~\cite{Manurangsi19-sosa} is the use of uniform random sampling of subsets of given cardinality; unfortunately, in the matroid setting, near-uniform sampling of independent sets is in general impossible. Instead, here we introduce the concept of \emph{Density-Balanced Subset} (DBS, Definition~\ref{def:bounded-density}) in matroids. We show that in those particular subsets of the ground set we can generalize the traditional notion of uniform random sampling of a subset of a given cardinality, to that of sampling a maximum independent set, while guaranteeing that (i) every element of the DBS has the same probability of being sampled (\emph{i.e.}, the probabilities are ``balanced'') and that (ii) those events are negatively correlated (Proposition~\ref{prop:random-sampling}).

    \paragraph*{Density-Balanced Subsets} We introduce here the concept of \emph{Density-Balanced Subsets}. Let us first define the notion of density. (In the following we assume that readers already have some familiarity with matroids; an introduction to matroids is provided in the beginning of Section~\ref{sec:density}.)

    \begin{restatable}[]{definition}{densitydefinition}
    \label{def:density}
        Let $\mathcal{M} = (V, \mathcal{I})$ be a matroid. The \emph{density} of a subset $U \subseteq V$ in $\mathcal{M}$ is defined as \[\rho_{\mathcal{M}}(U) = \frac{|U|}{\rnk_{\mathcal{M}}(U)}.\]
        The density of an empty set is set to $0$, and the density of a non-empty set of rank $0$ is $+\infty$.
    \end{restatable}
    
    From that we define a \emph{Density-Balanced Subset} (DBS). Basically, in a DBS, no subset has a larger density than the DBS itself, just like a \emph{uniformly dense} matroid, but here we add the constraint that the density has to be an integer.
    
    \begin{restatable}[]{definition}{boundeddensitydefinition}
    \label{def:bounded-density}
        Let $\mathcal{M} = (V, \mathcal{I})$ be a matroid, and $\rho$ be a positive integer. A subset $V' \subseteq V$ is called a \emph{$\rho$-DBS} in $\mathcal{M}$ if $\rho_{\mathcal{M}}(V') = \rho$ and for all $U \subseteq V'$, $\rho_{\mathcal{M}}(U) \leq \rho$.
    \end{restatable}

    Density-Balanced Subsets appear naturally when extracting independent sets in matroid unions (as we will see in Section~\ref{sec:kernel}); in Figure~\ref{fig:laminar-matroid} a simple example of DBS is given.

    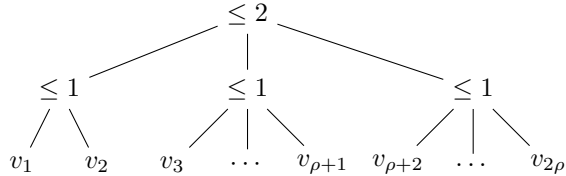
\begin{figure}[h]
    	\centering
        \begin{tikzpicture}
            \node (A1) at (1,0) {$v_1$};
            \node (A2) at (2,0) {$v_2$};
            
            \node (B1) at (3,0) {$v_3$};
            \node (B2) at (4,0) {$\cdots$};
            \node (B3) at (5,0) {$v_{\rho + 1}$};
            
            \node (C1) at (6,0) {$v_{\rho + 2}$};
            \node (C2) at (7,0) {$\dots$};
            \node (C3) at (8,0) {$v_{2\rho}$};
            
            \node (A) at (1.5,1) {$\leq 1$};
            \draw [-] (A) -- (A1);
            \draw [-] (A) -- (A2);

            \node (B) at (4,1) {$\leq 1$};
            \draw [-] (B) -- (B1);
            \draw [-] (B) -- (B2);
            \draw [-] (B) -- (B3);
            
            \node (C) at (7,1) {$\leq 1$};
            \draw [-] (C) -- (C1);
            \draw [-] (C) -- (C2);
            \draw [-] (C) -- (C3);
            
            \node (X) at (4,2) {$\leq 2$};
            \draw [-] (X) -- (A);
            \draw [-] (X) -- (B);
            \draw [-] (X) -- (C);
            
        \end{tikzpicture}
        \caption{Example of $\rho$-DBS of rank $k = 2$. The tree represents a laminar matroid $\mathcal{M} = (V, \mathcal{I})$ on the ground set $V = \{v_1, \dots, v_{2\rho}\}$: the leaves represent elements of the ground set, and the inner nodes represent cardinality constraints on the elements in their associated subtree (\emph{e.g.}, if $S \in \mathcal{I}$, then $|S \cap \{v_3, \dots, v_{\rho + 1}\}| \leq 1$). Observe that $V' = V$ is a $\rho$-DBS.}
        \label{fig:laminar-matroid}
    \end{figure}
    
    In $\rho$-DBSes, it is possible to sample independent sets while having the desired balance and negative correlation properties. Moreover, the sampled independent sets are ``maximal''  (\emph{i.e.}, these sampled independent sets are bases in the restriction of the matroid to the DBS), hence, in this sense, this extends the notion of uniform random sampling. This result comes from a more general rounding algorithm developed by Chekuri, Vondrák, and Zenklusen~\cite{ChekuriVZ10}.
    
    \begin{restatable}[]{proposition}{propsampling}
    \label{prop:random-sampling}
        Let $\mathcal{M} = (V, \mathcal{I})$ be a matroid, $V' \subseteq V$ be a $\rho$-DBS for some positive integer $\rho$, and $k = \rnk_{\mathcal{M}}(V')$. There exists a procedure to sample randomly from $V'$ an independent set of $k$ elements $S = \{s_1, \dots, s_k\} \in \mathcal{I}$ such that:
        \begin{enumerate}[(i)]
            \item for all $v \in V'$, $\mathbb{P}[v \in S] = 1/\rho$;
            \item for all $T \subseteq V'$, $\mathbb{P}[T \subseteq S] \leq \prod_{v \in T} \mathbb{P}[v \in S] = 1/\rho^{|T|}$.
        \end{enumerate}
    \end{restatable}

    \paragraph*{Matroid-Constrained Maximum Coverage} Let $\mathcal{M} = (V, \mathcal{I})$ be a matroid of rank $k$ on a ground set $V$. Given a universe $U$, a weight function $w: U \rightarrow \mathbb{R}_+$, and a family $\{U_v\}_{v \in V}$ of subsets of $U$, the \emph{matroid-constrained maximum coverage} problem is to select a subset $S \in \mathcal{I}$ maximizing the coverage function $f(S) = w(\bigcup_{s \in S}U_s) = \sum_{u \in \bigcup_{s \in S}U_s}w(u)$, namely, to find an element in
    \[\mathop{\arg\max}_{S \in \mathcal{I}} f(S).\]
    The \emph{frequency} of an element of the universe $u \in U$ refers to the number of sets $U_v$ in which it appears. We say that $f$ has \emph{bounded frequency} $\mu$ if every element of the underlying universe $U$ has a frequency bounded by $\mu$. In our paper, we will focus on the matroid-constrained maximum coverage problem for coverage functions having some bounded frequency $\mu$; this assumption is quite common and is used for instance in~\cite{BonnetPS16, FeldmannSLM20, HuangS22-swat, McGregor21-cover, SkowronF15, Skowron17}. An important special case is $\mu = 2$, where it corresponds to the coverage function over edges in a graph. Besides the frequency parameter, another parameter $z$, corresponding to the number of points covered in an optimal solution, has been used to design FTP algorithms in~\cite{HuangW2020}.
    
    To put our problem in a larger picture, we can first observe that a coverage function is a special case of monotone submodular function. For the problem of maximizing a monotone submodular function under a matroid constraint, an approximation of $1-1/e$ can be achieved in polynomial time~\cite{CalinescuCPV11} using continuous greedy and pipage rounding techniques. Later, a combinatorial approach for maximizing coverage functions over matroids was developed to achieve the same ratio~\cite{FilmusW12-cover}, and this approach was then generalized to monotone submodular functions~\cite{FilmusW14-submod}. 
    
    A special case of our problem when the matroid constraint is simply a cardinality constraint (\emph{i.e.}, a uniform matroid) has been studied in~\cite{AgeevS04, Feige98, Hochbaum1998,  Srinivasan01}. It has been shown in~\cite{Hochbaum1998} that a simple greedy procedure (picking at each step the element maximizing the increase of the coverage function) guarantees a ratio of $1-1/e$. If a polynomial time algorithm could approximate maximum coverage within a ratio of $1 - 1/e + \varepsilon$ for some $\varepsilon > 0$, then it would imply that $P = NP$~\cite{Feige98}. Furthermore, one cannot obtain in FPT time (where the matroid rank $k$ is the parameter) an approximation ratio better than $1-1/e+\varepsilon$, assuming GAP-ETH~\cite{Manurangsi20}. However, when the coverage function has bounded frequency $\mu$, an approximation of $1 - (1 - 1/\mu)^{\mu}$ can be achieved~\cite{AgeevS04}. 
    
    
    The case where the coverage function has a frequency $\mu$ bounded by $2$ is called the \textsc{matroid-constrained vertex cover}~\cite{HuangS22-swat,HuangS23}, and is called \textsc{max $k$-vertex-cover} when the matroid is uniform. The latter has also been studied through the lens of fixed-parameterized-tractability. The problem is $W[1]$-hard with $k$ being the parameter~\cite{GuoNW05}, thus getting an exact solution in FPT time is unlikely. Nonetheless, it is possible to get a near-optimal solution in FPT time~\cite{Marx08}. Precisely, an FPT approximation scheme (FPT-AS) is given in~\cite{Marx08}, that delivers a $1-\varepsilon$ approximate solution in $(k/\varepsilon)^{O(k^3/\varepsilon)} \cdot |V|^{O(1)}$ time, later improved to $(1/\varepsilon)^{O(k)} \cdot |V|^{O(1)}$ in~\cite{Manurangsi19-sosa, SkowronF15} .
    
    Here we recall the definition of an FPT-AS, introduced by Marx~\cite{Marx08}:
    \begin{definition}
        Given a parameter function $\kappa$ associating a positive integer to each instance $x \in I$ of some problem, a \emph{Fixed-Parameter Tractable Approximation Scheme (FPT-AS)} is an algorithm that provides a $(1-\varepsilon)$ approximate solution in time $g(\varepsilon, \kappa(x)) \cdot |x|^{O(1)}$ for some computable function $g$.
    \end{definition}
    In our case, each of the instances consists of a bounded-frequency coverage function and a matroid, and the parameter of an instance is the rank $k$ of its matroid.
    
    For our problem, an FPT-AS has been designed for partition, laminar, and transversal matroids in~\cite{HuangS22-swat}, where the concept of \emph{robust subset} is introduced to generalize the random sampling argument developed in~\cite{Manurangsi19-sosa}. In~\cite{HuangS22-swat}, an approximate kernel\footnote{An $\alpha$-approximate kernel~\cite{LokshtanovPRS17, Manurangsi19-sosa} for some parameterized optimization problem is a pair of polynomial time algorithms $\mathcal{A}$, the \emph{reduction algorithm}, and $\mathcal{B}$, the \emph{solution lifting algorithm}, such that (i) given an instance $(x, \kappa(x))$, $\mathcal{A}$ produces another instance $(x', \kappa(x'))$ such that $|x'|$, $\kappa(x')$ are bounded by $g(\kappa(x))$ and (ii) given a $\beta$ approximate solution $S'$ for $(x', \kappa(x'))$, $\mathcal{B}$ produces a solution $S$ of $(x, \kappa(x))$ such that $S$ is an $\alpha\beta$ approximate solution for $(x, \kappa(x))$. In the following we will use the terms ``kernel'' and ``approximate kernel'' interchangeably, dropping the adjective.} is extracted, consisting of a maximum weight independent set in the union of several copies of the same matroid $\mathcal{M}$ (see below for a definition of the union of matroids), and then a brute-force enumeration is performed on that kernel of small size. The number of matroids in these unions depends on the type of matroid.
    
    \begin{restatable}[]{definition}{definitionmatroidunion} \label{def:matroid-union}
	    Let $\mathcal{M}=(V, \mathcal{I})$ be a matroid. Then we can define $\rho \mathcal{M} = (V, \mathcal{I}_{\rho})$ as the \emph{union of $\rho$ matroids $\mathcal{M}$}, as follows: $S \in \mathcal{I}_{\rho}$ if $S$ can be partitioned into $S_1\cup \cdots \cup S_{\rho}$ so that for all $i$ we have $S_i \in \mathcal{I}$.
	\end{restatable}
	
	It is known that the union of matroids is still a matroid and that an independence oracle for $\rho \mathcal{M}$ can be implemented in polynomial time given an independence oracle for $\mathcal{M}$, \emph{e.g.}, see~\cite{Sch2003}. Moreover, the rank of $\rho \mathcal{M}$ is at most $\rho$ times the rank of $\mathcal{M}$.
	
	In this paper, we prove that a maximum weight independent set in $\rho \mathcal{M}$ contains an approximate solution of the problem, and that this ratio does not depend on the type of matroid (unlike~\cite{HuangS22-swat}):
    
    \begin{restatable}[]{theorem}{theoremunion} \label{thm:ratio-union}
        Let $\mathcal{M} = (V , \mathcal{I})$ be a matroid and let $f$ be a coverage function on $V$ of frequency bounded by $\mu$. Let $V'$ be a maximum weight independent set in $\rho \mathcal{M}$, with respect to the weights $f(\{v\})$. Then $V'$ contains a $1 - (\mu - 1)/\rho$ approximate solution of the matroid-constrained maximum coverage problem.
    \end{restatable}

    The proof of Theorem~\ref{thm:ratio-union} relies on a reinterpretation of the greedy algorithm (extracting $V'$ in the matroid $\rho\mathcal{M}$, \emph{i.e.}, Algorithm~\ref{algo:simple-kernel-construction}) as the process of constructing $\rho$-DBSes in a series of contracted matroids of $\mathcal{M}$ (see Algorithm~\ref{algo:kernel-construction}).
    These DBSes that appear during the construction of the kernel can then be used for random sampling purposes, allowing us to generalize the argument of~\cite{Manurangsi19-sosa} for uniform matroids to any matroid (details in Section~\ref{sec:kernel}). 
    
    \begin{remark}
        The simplicity of our characterization of the kernel implies that our kernelization process can be easily turned into a streaming algorithm, as in~\cite{HuangS22-swat}. In fact, assuming that the matroid as well as the cover function is provided to the algorithm as oracles, maintaining a maximum weight independent set in $\rho \mathcal{M}$ with respect to the weights $f(\{v\})$ can be done in streaming use $O(\rho \cdot k)$ memory.
    \end{remark}
    
    By taking the appropriate value $\rho = \lceil (\mu - 1)/\varepsilon\rceil$ and performing a bruteforce enumeration on the approximate kernel described in Theorem~\ref{thm:ratio-union} (that kernel would be of size $\rho \cdot k = O(k \cdot \mu/\varepsilon)$ and could be extracted in polynomial time (assuming that an independence oracle for $\mathcal{M}$ and an oracle for $f$ are given), we obtain an FPT-AS, extending the result of~\cite{Manurangsi19-sosa}:
    
    \begin{corollary} \label{cor:fpt-as}
        There exists an algorithm that computes a $1 - \varepsilon$ approximate solution of the matroid-constrained maximum coverage in $(\mu/\varepsilon)^{O(k)} \cdot |V|^{O(1)}$ time.
    \end{corollary}

    To see the interest of this result, we note that if the only parameter is the rank $k$ of the matroid, one cannot obtain in FPT time an approximation ratio better than $1-1/e+\varepsilon$ (assuming GAP-ETH) for a general coverage function~\cite{Manurangsi20}, even if the matroid is the simplest uniform matroid. In contrast, our result shows that it is possible to break through this lower bound and get arbitrarily close to $1$ for an arbitrary matroid when the frequency of the coverage function is bounded. We also emphasize here that the randomization is only used in the analysis; our algorithm itself is deterministic.
    
    
    An important special case of Theorem~\ref{thm:ratio-union} is when $\mu = 2$, which corresponds to the matroid-constrained maximum vertex cover problem, and for which we have:
    
    \begin{corollary}
        Let $\mathcal{M} = (V , \mathcal{I})$ be a matroid and let $G = (V, E)$ be a weighted graph. Let $V'$ be a maximum weight independent set in $\rho \mathcal{M}$, with respect to the weighted degrees $\deg_w(v)$. Then $V'$ contains a $1 - 1/\rho$ approximate solution of the matroid-constrained maximum vertex cover problem.
    \end{corollary}
    
    This extends and improves the previous kernelization results for this problem~\cite{HuangS22-swat, Manurangsi19-sosa}. In fact, in~\cite{Manurangsi19-sosa} the $1 - 1/\rho$ approximation is attained for the union of $\rho$ uniform matroids, and in~\cite{HuangS22-swat} that ratio is attained either for the union of $\rho$ partition matroids, $2 \cdot \rho$ laminar matroids, or $\rho + k - 1$ (reduced later to $\rho$ in~\cite{Kamiyama22-robust}) transversal matroids.

\section{Density-Balanced Subsets}
    \label{sec:density}
    
    Let us start with some definitions. Given a finite set $V$, a matroid is a pair $\mathcal{M} = (V, \mathcal{I})$ where $\mathcal{I} \subseteq \mathcal{P}(V)$ is a family of subsets in $V$ that satisfies the following three conditions:
    \begin{enumerate}[(1)]
        \item $\emptyset \in \mathcal{I}$, 
        \item if $X\subseteq Y \in \mathcal{I}$, then $X\in \mathcal{I}$,
        \item if $X, Y \in \mathcal{I}, |Y| > |X|$, then there exists an element $e \in Y \backslash X$ so that $X \cup \{e\} \in \mathcal{I}$.
    \end{enumerate}

    The set $V$ is called the \emph{ground set} of the matroid and the elements of $\mathcal{I}$ are called the \emph{independent sets}.
    The notion of matroid clearly generalizes that of linear independence in vector spaces; moreover, it also generalizes that of graph: in a \emph{graphic matroid}, the edges of the graph form the ground set and the independent sets are the acyclic sets of edges, namely, the forests. 
    
    Matroids terminology borrows concepts from vector spaces as well as graph theory. The \emph{rank} of a subset $X \subseteq V$ is defined as $\rnk_{\mathcal{M}}(X) = \max_{Y \subseteq X,\,Y \in \mathcal{I}}|Y|$; the rank of a matroid is defined as $\rnk_{\mathcal{M}}(V)$. The \emph{span} of a subset $X \subseteq V$ in the matroid $\mathcal{M}$ is defined as $\spn_{\mathcal{M}}(X) = \{x \in V : \rnk_{\mathcal{M}} (X \cup \{x\}) = \rnk_{\mathcal{M}}(X)\}$, and these elements in the \emph{span} are called \emph{spanned by $X$} in $\mathcal{M}$. A subset $C \subseteq V$ is a \emph{circuit} if $C$ is a minimal non-independent set, \emph{i.e.}, for every $v \in C$, $C \backslash \{v\} \in \mathcal{I}$. An element in $V$ that is a circuit by itself is called a \emph{loop}. For more details about matroids, we refer the reader to~\cite{Sch2003}.
    
    We recall the definition of a \emph{restriction} and a \emph{contraction} of a matroid. Performing such operation on a matroid results in another matroid.
    
    \begin{definition}[Restriction]
         Let $\mathcal{M} = (V, \mathcal{I})$ be a matroid, and let $V' \subseteq V$ be a subset. Then we define the \emph{restriction of $\mathcal{M}$ to $V'$} as $\mathcal{M}|V' = (V', \mathcal{I}')$ where $\mathcal{I'} = \{S \subseteq V' : S \in \mathcal{I}\}$
    \end{definition}
    
    \begin{definition}[Contraction]
        Let $\mathcal{M} = (V, \mathcal{I})$ be a matroid, and let $U$ be a subset of $V$. Then we define the \emph{contracted matroid $\mathcal{M}/U = (V \backslash U, \mathcal{I}_U)$} so that, given a maximum independent subset $\mathcal{B}_U$ of $U$, $\mathcal{I}_U = \{S \subseteq V \backslash U : S \cup \mathcal{B}_U \in \mathcal{I}\}$.
    \end{definition}
    
    It is well-known that any choice of $\mathcal{B}_U$  produces the same $\mathcal{I}_U$, as a result the definition of contraction is unambiguous. 
    The following proposition comes directly from the definition. 
	\begin{proposition} \label{prop:rank-contraction}
		Let $\mathcal{M} = (V, \mathcal{I})$ be a matroid and let $A \subseteq B \subseteq V$. Then we have $\rnk_{\mathcal{M}/A}(B\backslash A)  = \rnk_{\mathcal{M}}(B) - \rnk_{\mathcal{M}}(A)$.
	\end{proposition}
	
	In this paper, we will use the notion of density of a subset in a matroid. We recall the definition. We can observe from that definition that the density of a non-empty set is always larger or equal to one.
    \densitydefinition*

    We now introduce the notion of \emph{Density-Balanced Subset} (DBS).
    \boundeddensitydefinition*
    
    Here is a theorem due to Edmonds~\cite{edmonds1965minimumpartition} that will allow us to get another characterization of Density-Balanced Subsets.
    
    \begin{theorem}[Theorem 1 in~\cite{edmonds1965minimumpartition}] \label{thm:edm-partition}
        The elements of a matroid $\mathcal{M}$ can be partitioned into as few as $\rho$ sets, each of which is independent, if and only if there is no subset $A$ of elements of $\mathcal{M}$ for which $|A| > \rho \cdot \rnk_{\mathcal{M}}(A)$.
    \end{theorem}
    
    \begin{proposition} \label{prop:dbs-partition}
        Let $\mathcal{M} = (V, \mathcal{I})$ be a matroid. If $V'$ is a $\rho$-DBS for some positive integer $\rho$, then there exist $\rho$ independent sets $B_1, \dots, B_{\rho}$ such that $V' = B_1 \cup \dots \cup B_{\rho}$. Conversely, if a set $V'$ of density $\rho$ can be partitioned into $\rho$ independent sets $B_1, \dots, B_{\rho}$, then $V'$ is a $\rho$-DBS.
    \end{proposition}
    
    \begin{proof}
        Consider the matroid $\mathcal{M}|V'$, \emph{i.e.}, the restriction of $\mathcal{M}$ to $V'$. By Theorem~\ref{thm:edm-partition}, as the density of any set $A \subseteq V'$ is bounded by $\rho$, $V'$ can be partitioned into $\rho$ independent sets $B_1, \dots, B_{\rho}$. 
        
        Conversely, if $V'$ can be partitioned into $\rho$ independent sets $B_1, \dots, B_{\rho}$, then by Theorem~\ref{thm:edm-partition} the density of any set $A \subseteq V'$ is bounded by $\rho$. As we assumed $\rho_{\mathcal{M}}(V') = \rho$, $V'$ is a $\rho$-DBS.
    \end{proof}

    In DBSes, the notion of uniform random sampling can be properly extended.
    
    \propsampling*
    
    \begin{proof}
        By Proposition~\ref{prop:dbs-partition}, we can write $V' = B_1 \cup \dots \cup B_{\rho}$ for some disjoint independent sets.
        Hence we have (denoting $\mathbbm{1}_{U}$ the indicator vector of the set $U$):\[\frac{1}{\rho} \mathbbm{1}_{V'} = \frac{1}{\rho}\sum_{i = 1}^{\rho} \mathbbm{1}_{B_i} \in P(\mathcal{M}),\]
        where \[P(\mathcal{M}) = \textrm{conv}\{\mathbbm{1}_{S} : S \in \mathcal{I}\} = \{x \in [0,1]^{V} : \forall\,S \subseteq V, \sum_{v \in S}x_v \leq \rnk_{\mathcal{M}}(S)\}\]
        denotes the matroid polytope of $\mathcal{M}$. Therefore we can apply the randomized rounding algorithm developed in~\cite{ChekuriVZ10} to the vector $\frac{1}{\rho} \mathbbm{1}_{V'}$ to get an integral vector $X = \mathbbm{1}_S \in \{0,1\}^V$ such that $S \subseteq V'$ is an independent set and
        \begin{enumerate}[(i)]
            \item for all $v \in V'$, $\mathbb{P}[v \in S] = \mathbb{E}[X_v] = 1/\rho$;
            \item for all $T \subseteq V'$, $\mathbb{P}[T \subseteq S] = \mathbb{E}[\prod_{v \in T}X_v] \leq \prod_{v \in X} \mathbb{E}[X_v] = \prod_{v \in T} \mathbb{P}[v \in S] = 1/\rho^{|T|}$.
        \end{enumerate}
        This concludes the proof.
    \end{proof}

    Conversely we also have the following proposition.
    \begin{proposition}
        Let $V'$ a set of integer density $\rho$ in which independent sets can be sampled randomly so that each element has probability $1/\rho$ of being sampled, then $V'$ is a $\rho$-DBS.
    \end{proposition}

    \begin{proof}
        In fact, if there exists $U \subseteq V'$ such that $\rho_{\mathcal{M}}(U) > \rho$, then an algorithm sampling each element in $V'$ with probability $1/\rho$ would take in expectation strictly more than $\rnk_{\mathcal{M}}(U)$ elements in $U$, meaning that some of those sampled set violate the rank constraint on $U$.
    \end{proof}
    
    Another useful property of DBSes is that after a matroid contraction we can still recover in a $\rho$-DBS $V'$ a smaller one $V''$ while preserving the rank of $V'$ in the contracted matroid:
    
    \begin{proposition} \label{prop:dbs-contraction}
        Let $\mathcal{M} = (V, \mathcal{I})$ be a matroid, and $V'$ be a $\rho$-DBS in $\mathcal{M}$ for some positive integer $\rho$. Let $A \subseteq V$ such that $V' \not\subseteq \spn_{\mathcal{M}}(A)$. Then there exists a subset $V'' \subseteq V'$ such that:
        \begin{enumerate}[(i)]
            \item $V''$ is a $\rho$-DBS in $\mathcal{M}/A$;
            \item $\rnk_{\mathcal{M}/A}(V'') = \rnk_{\mathcal{M}/A}(V' \backslash A)$.
        \end{enumerate}
    \end{proposition}
    
    \begin{proof}
        Using Proposition~\ref{prop:dbs-partition}, we know that $V' = B_1 \cup \dots \cup B_{\rho}$ for some disjoint independent sets $B_1, \dots, B_{\rho}$, each of cardinality $\rnk_{\mathcal{M}}(V')$ (because $V'$ has density $\rho$, so it contains $\rho \cdot \rnk_{\mathcal{M}}(V')$ elements, and each independent set is made of at most $\rnk_{\mathcal{M}}(V')$ elements), and each of them spanning $V'$ in $\mathcal{M}$. As a result, for all $i \in \{1, \dots, \rho\}$ we also have $V'  \backslash A \subseteq \spn_{\mathcal{M}/A}(B_i \backslash A)$, hence $\rnk_{\mathcal{M}/A}(B_i \backslash A) = \rnk_{\mathcal{M}/A}(V' \backslash A)$ and thereby there exists $B'_i \subseteq B_i \backslash A$ such that $B'_i$ is independent in $\mathcal{M}/A$ and $|B'_i| = \rnk_{\mathcal{M}/A}(V'\backslash A)$. Hence $V'' = B'_1 \cup \dots \cup B'_{\rho}$ is a set of rank equal to $\rnk_{\mathcal{M}/A}(V' \backslash A)$ in $\mathcal{M}/A$, contains $\rho \cdot \rnk_{\mathcal{M}/A}(V' \backslash A)$ elements, and can be partitioned into $\rho$ independent sets, therefore, by Proposition~\ref{prop:dbs-partition}, $V''$ is a $\rho$-DBS in $\mathcal{M}/A$ with $\rnk_{\mathcal{M}/A}(V'') = |B'_1| = \dots = |B'_{\rho}| = \rnk_{\mathcal{M}/A}(V' \backslash A)$.
    \end{proof}

    This contraction property of DBSes will be useful for our result in Section~\ref{sec:kernel}, and is the property of DBSes that requires the rank $\rho$ to be an integer. Moreover, this property also allow to design algorithms picking sequentially elements as demonstrated in Appendix~\ref{sec:random-seq}. Interestingly, that sampling technique specific to DBSes is very different perspective from the randomized rounding techniques described~\cite{ChekuriVZ10}.

    The next proposition states how the density is changed after a matroid is contracted.
    \begin{proposition}
        \label{prop:subset_increased_rho}
        Let $\mathcal{M} = (V, \mathcal{I})$ be a matroid. If $A \subseteq B \subseteq V$ and $U \subseteq V \backslash B$ we have the following inequality:
        \[\rho_{\mathcal{M}/A}(U) \leq \rho_{\mathcal{M}/B}(U),\]
        assuming that $\rho_{\mathcal{M}/A}(U) < +\infty$.
    \end{proposition}
    
    \begin{proof}
        In fact, $\rnk_{\mathcal{M}/A}(U) \geq \rnk_{\mathcal{M}/B}(U)$, while the cardinality $|U|$ remains obviously the same.
    \end{proof}
    
    Now we give some results regarding densest subsets, which are closely related to \emph{Density-Balanced Subsets}, as a densest subset in a matroid is automatically a DBS.
    
    \begin{proposition} \label{prop:add-one-densest}
        Let $\mathcal{M} = (V, \mathcal{I})$ be a matroid, and $\rho$ be a positive integer. Let $V' \subseteq V$. If $\max_{U \subseteq V'}\rho_{\mathcal{M}}(U) < \rho$, then for any $v \in V \backslash V'$, $\max_{U \subseteq V' \cup \{v\}}\rho_{\mathcal{M}}(U) \leq \rho$.
    \end{proposition}
    
    \begin{proof}
        Consider $U \subseteq V'$, $U \neq \emptyset$. As $\rho_{\mathcal{M}}(U) < \rho$, we know that $|U| \leq \rho \cdot \rnk_{\mathcal{M}}(U) - 1$ (because $\rho \cdot \rnk_{\mathcal{M}}(U)$ is an integer). Therefore we have
        \[\rho_{\mathcal{M}}(U \cup \{v\}) \leq \frac{|U| + 1}{\rnk_{\mathcal{M}}(U)} \leq \frac{\rho \cdot \rnk_{\mathcal{M}}(U) - 1 + 1}{\rnk_{\mathcal{M}}(U)} = \rho,\]
        which concludes the proof.
    \end{proof}
    
    \begin{proposition} \label{prop:largerdensity}
	    Let $\mathcal{M} = (V, \mathcal{I})$ be a matroid, $V'$ be a subset of $V$, and let $B$ be a subset that reaches the maximum density $\rho^* < +\infty$ in $V'$. Then given any $A \subsetneq B$, $\rho_{\mathcal{M}/ A}(B \backslash A) \geq \rho^{*}$. 
	\end{proposition}
	
	\begin{proof}
	    If $\rnk_{\mathcal{M}/A}(B \backslash A) = 0$ then $\rho_{\mathcal{M}/A}(B \backslash A) = +\infty$ and we are done; otherwise, by  Proposition~\ref{prop:rank-contraction}:
	    \[\rho_{\mathcal{M}}(B) = \frac{\rnk_{\mathcal{M}}(A) \cdot \rho_{\mathcal{M}}(A) + \rnk_{\mathcal{M}/A}(B \backslash A) \cdot \rho_{\mathcal{M}/A}(B \backslash A)}{\rnk_{\mathcal{M}}(A) + \rnk_{\mathcal{M}/A}(B \backslash A)},\]
	    hence $\rho_{\mathcal{M}}(B)$ is a weighted average of $\rho_{\mathcal{M}}(A)$ and $\rho_{\mathcal{M}/A}(B \backslash A)$. As $\rho_{\mathcal{M}}(A) \leq \rho^*$ (by definition of $\rho^*$), it implies that $\rho_{\mathcal{M}/A}(B \backslash A) \geq \rho^*$. 
	\end{proof}
    
    The following proposition states that the densest subsets are closed under union, thus proving the uniqueness of the maximum cardinality densest subset --- that property will be useful in Algorithm~\ref{algo:kernel-construction} (Section~\ref{sec:kernel}).
	
	\begin{proposition} \label{prop:unique-densest}
        Let $\mathcal{M} = (V, \mathcal{I})$ be a matroid, and $V' \subseteq V$. Let $\rho^* = \max_{U \subseteq V'}\rho_{\mathcal{M}}(U) < +\infty$. Then given any two subsets $W_1$, $W_2$ in $V'$ of density $\rho^*$, $\rho_{\mathcal{M}}(W_1 \cup W_2) = \rho^*$.
    \end{proposition}
    
    \begin{proof}
        If $W_1 \subseteq W_2$, then the proposition is trivially true. So assume that $W_1 \backslash W_2 \neq \emptyset$, and we can
        observe that
        \[\rho^* \leq  \rho_{\mathcal{M}/(W_1 \cap W_2)}(W_1 \backslash (W_1 \cap W_2)) \leq \rho_{\mathcal{M}/W_2}(W_1 \backslash (W_1 \cap W_2)),\]
        where the first inequality uses Proposition~\ref{prop:largerdensity} and the second uses Proposition~\ref{prop:subset_increased_rho}. 
        As a result, by the facts that $\rho_{\mathcal{M}}(W_2) = \rho^*$ and that  $\rho_{\mathcal{M}/W_2}(W_1 \backslash (W_1 \cap W_2)) \geq \rho^*$, we obtain $\rho_{\mathcal{M}}(W_1 \cup W_2)\geq \rho^*$ (using a weighted average argument as in Proposition~\ref{prop:largerdensity}).
        Hence we have $\rho_{\mathcal{M}}(W_1 \cup W_2) = \rho^*$.
    \end{proof}
    
    \begin{proposition} \label{prop:non-increasing}
	     Let $\mathcal{M} = (V, \mathcal{I})$ be a matroid, and $V' \subseteq V$ a non-empty set. Let $A$ be the largest densest subset in $V'$. Then for any $B \subseteq V' \backslash A$, we have $\rho_{\mathcal{M}/A}(B) < \rho_{\mathcal{M}}(A)$.
	\end{proposition}
	
	\begin{proof}
	    We proceed by contradiction. Suppose that there exists $B \subseteq V' \backslash A$ such that $\rho_{\mathcal{M}/A}(B) \geq \rho_{\mathcal{M}}(A)$. Then it implies that \[\rho_{\mathcal{M}}(A \cup B) = \rho_{\mathcal{M}}(A) \cdot \frac{\rnk_{\mathcal{M}}(A)}{\rnk_{\mathcal{M}}(A) + \rnk_{\mathcal{M}/A}(B)} + \rho_{\mathcal{M}/A}(B) \cdot \frac{\rnk_{\mathcal{M}/A}(B)}{\rnk_{\mathcal{M}}(A) + \rnk_{\mathcal{M}/A}(B)} \geq
	    \rho_{\mathcal{M}}(A),\]
	    contradicting the hypothesis that $A$ was the largest densest set in $V'$.
	\end{proof}
    
\section{Matroid-Constrained Maximum Coverage}
    \label{sec:kernel}

    Here we will use a slightly different formalization of the coverage function (using edge-weighted hypergraphs) compared to the one provided in the introduction, but it is straightforward to see that these formalizations are actually equivalent. Let $G = (V,E)$ be a hypergraph, $V$ being a set of vertices and $E$ being a set of hyper-edges, \emph{i.e.}, an element $e \in E$ is a subset of $V$. We denote $n = |V|$. Let $w : E \rightarrow \mathbb{R}_+$ be a weight function on the hyper-edges. We extend this function to any set of hyper-edges by setting for any $A \subseteq E$, $w(A) = \sum_{e \in A}w(e)$. For a vertex $v \in V$, we denote $\delta(v)$ the set of its set of \emph{incident} hyper-edges, namely, $\delta(v) = \{e \in E : v \in e\}$, and $\deg_w(v)$ its \emph{weighted degree}, namely, $\deg_w(v) = w(\delta(v))$. The \emph{frequency} of a hyper-edge $e$ is defined as the number of vertices for which $e$ appears in $\delta(v)$, namely, $|e|$. The hypergraph $G$ will be said of \emph{bounded frequency $\mu$} if all its hyper-edges have frequencies bounded by $\mu$. For two sets of vertices $S$, $T$ we denote by $E(S, T)$ the set of hyper-edges having at least one endpoint in each set $S$ and $T$, namely, $E(S,T) = \{e \in E : e \cap S \neq \emptyset, e \cap T \neq \emptyset\}$. For conciseness we will denote $E(S) = E(S, S)$.
    
    Let $\mathcal{M} = (V, \mathcal{I})$ be a matroid on the ground set $V$. In the matroid-constrained maximum coverage problem, we are asked to find a set of vertices $S \subseteq V$ that is independent in the matroid $\mathcal{M}$ (\emph{i.e.}, $S \in \mathcal{I}$) and that maximizes the total weight of the covered hyper-edges, namely, an element of
    \[\mathop{\arg\max}_{S \in \mathcal{I}} w(E(S)).\]
    The problem can be solved exactly by the standard greedy algorithm if $\mu = 1$~\cite{Edmonds1971}, so in the following we will assume that our hypergraph has a bounded frequency of $\mu \geq 2$. Observe that the case $\mu = 2$ corresponds to the matroid-constrained maximum vertex cover, studied in~\cite{HuangS22-swat}.
    
    Here we want to construct a kernel that contains a good approximation of the optimal solution of the maximum coverage problem under a matroid constraint. We will start by describing a procedure to build the kernel that will be convenient for the analysis, and then we will show that this algorithm turn out to be equivalent to the one of Theorem~\ref{thm:ratio-union}. We build our kernel $V'$ as follows. Let $\rho$ be a fixed positive integer. Start with an empty set $V'$, and an auxiliary set $C$ that is also empty at the beginning. Processing the elements $v_i \in V$ by non-increasing weighted-degree, if the element $v_i$ is not spanned by $V'$ at that time we add that element to $C$ and we check whether $C$ contains a set of density larger or equal to $\rho$ with respect to the matroid $\mathcal{M}/V'$. If this is the case, then we consider that largest densest subset $X$ in $C$ with respect to $\mathcal{M}/V'$, we add that set into $V'$ and remove that set from $C$ (note that, because of Proposition~\ref{prop:unique-densest}, the largest densest set $X$ is well-defined). When the main loop terminates, the set $C$ is also added into $V'$. A formal description of this procedure is provided in Algorithm~\ref{algo:kernel-construction}.
    
    \begin{algorithm}
	\caption{Algorithm for building a maximum coverage approximate kernel}\label{algo:kernel-construction}
	\begin{algorithmic}[1]
	\State $V = \{v_1,\dots, v_n\}$ where $\deg_w(v_1) \geq \dots \geq \deg_w(v_n)$
	\State $V' \gets \emptyset$, $C \gets \emptyset$
	\For{$i = 1, \dots, n$} \Comment the vertices are processed in non-increasing order of weighted degree
	    \If{$v_i \in \spn_{\mathcal{M}}(V')$} \label{line:kernel-dismiss}
	        \State\textbf{continue} \Comment $v_i$ is ignored if already spanned by $V'$
	    \EndIf
	    \State $C \gets C \cup \{v_i\}$
	    \State let $X$ be the largest densest subset in $C$ with respect to the matroid $\mathcal{M}/V'$
	    \If{$\rho_{\mathcal{M}/V'}(X) \geq \rho$} \label{line:kernel-bigger}
	        \State $C \gets C \backslash X$
	        \State $V' \gets V' \cup X$ \label{line:add-densest-kernel}
	    \EndIf
	\EndFor
	\State $V' \gets V' \cup C$ \label{line:add-end-kernel}
	\State\Return $V'$
	\end{algorithmic}
	\end{algorithm}
	
	\begin{claim} \label{claim:line-ineg-kernel}
        In Algorithm~\ref{algo:kernel-construction}, at the end of each iteration of the loop, for all $U \subseteq C$, $\rho_{\mathcal{M}/V'}(U) < \rho$. Moreover, when the condition at Line~\ref{line:kernel-bigger} is true, we have $\rho_{\mathcal{M}/V'}(X) = \rho$.
    \end{claim}
    
    \begin{proof}
        We prove these two properties by induction. In the first iteration of the algorithm, both properties are clearly true. Suppose that during the $i$th iteration both properties are satisfied. It means that at the beginning of the $(i+1)$st iteration the densest subset in $C$ is of density strictly smaller than $\rho$, hence Proposition~\ref{prop:add-one-densest} implies that the densest subset after inserting $v_{i+1}$ in $C$ cannot be of density strictly larger than $\rho$. This implies the second property for the $(i+1)$st iteration. Regarding the first property,
        \begin{itemize}
            \item if condition at Line~\ref{line:kernel-bigger} is false, then the first property is clearly satisfied;
            \item if condition at Line~\ref{line:kernel-bigger} is true, then the preceding discussion implies that the largest densest subset $X$ removed from $C$ has density $\rho$, and therefore, by Proposition~\ref{prop:non-increasing}, the first property is satisfied.
        \end{itemize}
        This concludes the proof.
    \end{proof}
    
    The set $V'$ can be decomposed as follows:
    \begin{equation} \label{eq:kernel-decomp}
        V' = X_1 \cup \dots \cup X_r \cup R
    \end{equation}
    where the $X_i$s represent the largest densest subsets $X$ that were added through the execution of the algorithm (at Line~\ref{line:add-densest-kernel}), labeled in the order they were added, and $R$ represents the set of remaining elements coming from $C$ that were added after termination of the main loop (at Line~\ref{line:add-end-kernel}). As each $X_i$ is a densest subset of density $\rho$ in $\mathcal{M}/(\bigcup_{j = 1}^{i-1}X_j)$, each set $X_i$ is a $\rho$-DBS in the matroid $\mathcal{M}/(\bigcup_{j = 1}^{i-1}X_j)$.
    
    Now we can prove the following important lemma, which states that the set $V'$ built by Algorithm~\ref{algo:kernel-construction} is an approximate kernel, \emph{i.e.}, a small subset of $V$ containing a good approximation of the optimal solution:
	
	\begin{lemma} \label{lem:ratio-kernel}
	    Let $V'$ be the kernel built in Algorithm~\ref{algo:kernel-construction}, and let $O$ be an optimal solution. Then $V'$ contains an independent set $S$ such that $w(E(S)) \geq (1 -  (\mu - 1)/\rho) \cdot w(E(O))$.
	\end{lemma}
	
	\begin{proof}
	    Let $O \in \mathcal{I}$ be an optimal solution. We denote $O^{in} = O \cap V'$, $O^{out} = O \backslash O^{in}$. As in~\cite{HuangS22-swat,Manurangsi19-sosa}, we want to sample randomly an independent set $S \subseteq V'$ so that we have an inequality \emph{in expectation}, implying that some set satisfying that same inequality actually exists:\[\mathbb{E}[w(E(S))] \geq (1 -  (\mu - 1)/\rho) \cdot w(E(O)).\]
	    
	    To sample $S$, we will do the following. For $i = 1, \dots, r$, using Proposition~\ref{prop:dbs-contraction}, consider $X'_i \subseteq X_i$ a $\rho$-DBS in $\mathcal{M}/(O^{in} \cup \bigcup_{j = 1}^{i-1}X_j)$ of rank $k_i = \rnk_{\mathcal{M}/(O^{in} \cup \bigcup_{j = 1}^{i-1}X_j)}(X_i)$ (if $X_i$ is spanned by $O^{in} \cup \bigcup_{j = 1}^{i-1}X_j$ in $\mathcal{M}$, then we just set $X'_i = \emptyset$ and $k_i = 0$). From that set $X_i'$, using Proposition~\ref{prop:random-sampling}, we sample an independent set $S_i = \{s_{i,1},\dots, s_{i,k_i}\} \subseteq X'_i$ in $\mathcal{M}/(O^{in} \cup \bigcup_{j = 1}^{i-1}X_j)$ such that:
	    \begin{enumerate}[(i)]
            \item for all $v \in X_i'$, $\mathbb{P}[v \in S_i] = 1/\rho$;
            \item for all $v, v' \in X_i'$ such that $v \neq v'$, $\mathbb{P}[v \in S_i \wedge v' \in S_i] \leq 1/\rho^2$.
        \end{enumerate}
        Then we set $S = O^{in} \cup \bigcup_{i = 1}^{r} S_i$. We denote $\tilde{S} = \bigcup_{i = 1}^{r} S_i$ and $V'' = \bigcup_{i = 1}^r X'_i$. The $S_i$s are sampled independently.
        
        \begin{claim} \label{claim:s-properties}
            The set $S$ sampled using the aforementioned method is always independent in $\mathcal{M}$. Moreover,
            \begin{enumerate}[(i)]
                \item for all $v \in V''$, $\mathbb{P}[v \in \tilde{S}] = 1/\rho$;
                \item for all $v, v' \in V''$ such that $v \neq v'$, $\mathbb{P}[v \in \tilde{S} \wedge v' \in \tilde{S}] \leq 1/\rho^2$.
        \end{enumerate}
        \end{claim}
        
        \begin{proof}
            We prove by induction on $i$ that $O^{in} \cup \bigcup_{j=1}^iS_j$ is independent in $\mathcal{M}$. For $i = 0$ this is clearly true, as $O^{in}$ is a subset of $O$, which is an independent set in $\mathcal{M}$. Suppose the property is true for some $i < r$. Then it means that $O^{in} \cup \bigcup_{j=1}^iS_j$ is an independent set in $\mathcal{M}$. We know that for any sampling in $X'_{i+1}$, the set $S_{i+1}$ is independent is $\mathcal{M}/(O^{in} \cup \bigcup_{j=1}^iX_j)$, so it is also independent in $\mathcal{M}/(O^{in} \cup \bigcup_{j=1}^iS_j)$ (as $O^{in} \cup \bigcup_{j=1}^iS_j \subseteq O^{in} \cup \bigcup_{j=1}^iX_j$) and therefore $O^{in} \cup \bigcup_{j=1}^{i+1}S_j$ is independent in $\mathcal{M}$.
            
            Then, for property~(i), consider some $v \in V''$. There exists a unique $i$ such that $v \in X_i$ (as the $X_i$s are disjoint). Hence from the properties of the sampling of $S_i$ we know that $\mathbb{P}[v \in S_i] = 1/\rho$, hence $\mathbb{P}[v \in \tilde{S}] = 1/\rho$. For property~(ii), if $v$ and $v'$ are in the same $X_i$, then the way $S_i$ is sampled guarantees that $\mathbb{P}[v \in \tilde{S} \wedge v' \in \tilde{S}] \leq 1/\rho^2$. Otherwise, the choices of $v$ and $v'$ are independent, \emph{i.e.}, $\mathbb{P}[v \in \tilde{S} \wedge v' \in \tilde{S}] = 1/\rho^2$.
        \end{proof}
        
        For $i = 1, \dots, r$, let $O^{out}_i = O^{out} \cap (\spn_{\mathcal{M}}(\bigcup_{j = 1}^{i}X_j) \backslash \spn_{\mathcal{M}}(\bigcup_{j = 1}^{i-1}X_j))$. 
        
        \begin{claim}
            We have $O^{out} = O^{out} \cap \spn_{\mathcal{M}}(\bigcup_{i = 1}^{r}X_i) = \bigcup_{i = 1}^{r}O^{out}_i$.
        \end{claim} 
        
        \begin{proof}
            The second part of the equality is straightforward, so we focus on the first part. Consider $v \in O \backslash \spn_{\mathcal{M}}(\bigcup_{i = 1}^{r}X_i)$. It means that when $v$ is processed in Algorithm~\ref{algo:kernel-construction}, that element cannot be discarded by the condition in Line~\ref{line:kernel-dismiss}, because at that time $V' = \bigcup_{j = 1}^{i}X_j$ for some $i$ and therefore $v$ is not spanned by $V'$: that element is thereby added to $C$. Hence $v$ is added to $V'$ in the end (Line~\ref{line:add-end-kernel}) and $v \in O^{in}$ (more precisely, $v \in O \cap R$, using the notation of equation~(\ref{eq:kernel-decomp})).
        \end{proof}
        
        \begin{claim} \label{claim:x-greater}
            For all $v \in O^{out}_i$, for all $v' \in X_j$ such that $j \leq i$, we have $\deg_w(v) \leq \deg_w(v')$.
        \end{claim}
        
        \begin{proof}
            The element $v \in O^{out}_i$ has been discarded after the sets $X_1, \dots, X_i$ were built (otherwise it would not have been spanned by $V'$, see Line~\ref{line:kernel-dismiss}), therefore these sets only contain elements having larger or equal weighted degrees.
        \end{proof}
        
        \begin{claim} \label{claim:ineg-k}
            For all $1 \leq i \leq r$, we have $|\bigcup_{j=1}^{i}O^{out}_j| \leq \sum_{j = 1}^{i}k_j$.
        \end{claim}
        
        \begin{proof}
            The set $\bigcup_{j=1}^{i}O^{out}_j$ is independent in the matroid $\mathcal{M}/O^{in}$ and as it is in $\spn_{\mathcal{M}}(\bigcup_{j = 1}^{i}X_j)$, it is in $\spn_{\mathcal{M}/O^{in}}(\bigcup_{j = 1}^{i}X_j)$. Then we have the inequality $|\bigcup_{j=1}^{i}O^{out}_j| \leq \rnk_{\mathcal{M}/O^{in}}(\bigcup_{j = 1}^{i}X_j) = \sum_{j=1}^{i}\rnk_{\mathcal{M}/(O^{in} \cup \bigcup_{l = 1}^{j-1}X_l)}(X_j) = \sum_{j = 1}^{i}k_j$, where in the first equality we used Proposition~\ref{prop:rank-contraction} multiple times.
        \end{proof}
        
        Now we index the elements in $O^{out} = \{o_1, \dots, o_{|O^{out}|}\}$ so that $O^{out}_1 = \{o_1, \dots,o_{|O^{out}_1|}\}$, $O^{out}_2 = \{o_{|O^{out}_1| + 1}, \dots ,o_{|O^{out}_1| + |O^{out}_2|}\}$, and so on. Similarly, we index the elements of $\tilde{S} = \{s_1, \dots, s_{\sum_{i=1}^rk_i}\}$ so that $s_1 = s_{1,1}, \dots, s_{k_1} = s_{1, k_1}$, $s_{k_1 + 1} = s_{2,1}, \dots, s_{k_1 + k_2} = s_{2,k_2}$, and so on.
        
        In the following, we will say that the element $s_i$ ``replaces'' the element $o_i$.
        
        \begin{claim} \label{claim:replace-greater}
            For all $1 \leq i \leq |O^{out}|$, we have $\deg_w(o_i) \leq \deg_w(s_i)$.
        \end{claim}
        
        \begin{proof}
            Because of Claim~\ref{claim:ineg-k}, an element $o_i \in O^{out}_j$ is replaced by $s_i \in X_{j'}$ for some $j' \leq j$. As a result, by Claim~\ref{claim:x-greater}, we know that we always have $\deg_w(o_i) \leq \deg_w(v)$ for any $v \in X_{j'}$. As $s_i$ is drawn from $X_{j'}$ we obtain the desired result.
        \end{proof}
    
        Then, as $S = O^{in} \cup \tilde{S}$, we have:
	    \[w(E(S)) = w(E(O^{in})) + w(E(\tilde{S})) - w(E(O^{in}, \tilde{S})).\]
	    We bound $\mathbb{E}[w(E(O^{in}, \tilde{S}))]$ as follows. By construction, $\mathbb{P}[v' \in V''] = 1/\rho$ for all $v' \in V''$. Then we have
        \begin{align*}
            \mathbb{E}[w(E(O^{in}, \tilde{S}))] &= \sum_{e \in E(O^{in})} w(e) \cdot \mathbb{P}[e \cap \tilde{S} \neq \emptyset]\\
            &\leq \sum_{e \in E(O^{in})} \left[w(e) \cdot \left(\sum_{v' \in e \cap V''}\mathbb{P}[v' \in \tilde{S}]\right)\right] &\text{by union-bound}\\
            &= \sum_{e \in E(O^{in})} w(e) \cdot |e \cap V''| \cdot 1/\rho &\text{as $\mathbb{P}[v' \in \tilde{S}] = 1/\rho$ for all $v' \in V''$}\\
            &\leq \sum_{e \in E(O^{in})} w(e) \cdot (\mu - 1) \cdot 1/\rho &\text{as $|e \cap V''| \leq \mu - 1$}\\
            &= (\mu - 1)/\rho \cdot w(E(O^{in})).
        \end{align*}
        Furthermore, the value $w(E(\tilde{S}))$ can be rearranged as follows:
        \begin{align*}
            w(E(\tilde{S})) = \sum_{e \in E(\tilde{S})} \sum_{v \in e \cap \tilde{S}} \frac{w(e)}{|e \cap \tilde{S}|} = \sum_{v \in \tilde{S}} \sum_{e \in \delta(v)} \frac{w(e)}{|e \cap \tilde{S}|} &=  \sum_{v \in \tilde{S}} \sum_{e \in \delta(v)}\left(\left(1 - \frac{|e \cap \tilde{S}| - 1}{|e \cap \tilde{S}|}\right) \cdot w(e)\right)\\
            &= \sum_{v \in \tilde{S}} \left(\deg_w(e) - \sum_{e \in \delta(v)}\frac{|e \cap \tilde{S}| - 1}{|e \cap \tilde{S}|} \cdot w(e)\right).
        \end{align*}
        Hence $\mathbb{E}[w(E(\tilde{S}))]$ can be written as
	    \begin{equation}\label{eq:eps-s-tilde}
	        \mathbb{E}[w(E(\tilde{S}))] = \mathbb{E}\left[\sum_{v \in \tilde{S}}\deg_w(v)\right] - \mathbb{E}\left[\sum_{v \in \tilde{S}}\sum_{e \in \delta(v)} \frac{|e \cap \tilde{S}| - 1}{|e \cap \tilde{S}|} \cdot w(e)\right].
	    \end{equation}
        We will then focus on upper-bounding the second term, which captures the extent to which edges are counted multiple times in the first term of the sum. We have
         \begin{align*}
	        \mathbb{E}&\left[\sum_{v \in \tilde{S}}\sum_{e \in \delta(v)} \frac{|e \cap \tilde{S}| - 1}{|e \cap \tilde{S}|} \cdot w(e)\right]\\
            &\leq \mathbb{E}\left[\sum_{v \in V''} \sum_{e \in \delta(v)} w(e) \cdot \mathbbm{1}[v \in \tilde{S} \wedge |e \cap \tilde{S}| \geq 2]\right]\\
	        &= \sum_{v \in V''} \sum_{e \in \delta(v)}  w(e) \cdot \mathbb{P}[v \in \tilde{S} \wedge |e \cap \tilde{S}| \geq 2]\\
	        &\leq \sum_{v \in V''}  \sum_{e \in \delta(v)}  w(e) \cdot \left(\sum_{v' \in e \cap V'' \backslash \{v\}}\mathbb{P}[v \in \tilde{S} \wedge v' \in \tilde{S}]\right) &\text{by union-bound}\\
	        &\leq \sum_{v \in V''} \sum_{e \in \delta(v)} w(e) \cdot \left(\sum_{v' \in e \cap V'' \backslash \{v\}}1/\rho^2\right) &\text{by Claim~\ref{claim:s-properties}}\\
            &\leq \sum_{v \in V''} \sum_{e \in \delta(v)} w(e) \cdot \left((\mu - 1) \cdot 1/\rho^2\right) &\text{as $|e \cap V'' \backslash \{v\}| \leq \mu - 1$}\\
	        &= (\mu - 1) \sum_{v \in V''} 1/\rho^2 \cdot \deg_w(v)\\
	        &= (\mu - 1)/\rho \cdot \mathbb{E}\left[\sum_{v \in \tilde{S}} \deg_w(v)\right]. &\text{as $\mathbb{P}[v \in \tilde{S}] = 1/\rho$ for all $v \in V''$}
	    \end{align*}
        where for the first inequality we apply the worst possible coefficient (namely, $1 \geq (\mu - 1)/\mu$) each time $e$ is covered more than once by $\tilde{S}$, and for the second inequality we use a union bound, namely, $\mathbb{P}[v \in \tilde{S} \wedge |e \cap \tilde{S}| \geq 2] \leq \sum_{v' \in e \cap V'' \backslash \{v\}}\mathbb{P}[v \in \tilde{S} \wedge v' \in \tilde{S}]$. As a result, combining with~(\ref{eq:eps-s-tilde}), we obtain 
        \[\mathbb{E}[w(E(\tilde{S}))] \geq (1 - (\mu - 1)/\rho) \cdot \mathbb{E}\left[\sum_{v \in \tilde{S}}\deg_w(v)\right].\]
	    Then this can be compared to $w(E(O^{out}))$ as we have
	    \[\mathbb{E}\left[\sum_{v \in \tilde{S}} \deg_w(v)\right]
	        \geq\, \mathbb{E}\left[\sum_{i = 1}^{|O^{out}|} \deg_w(s_i)\right]
	        \geq \sum_{i = 1}^{|O^{out}|} \deg_w(o_i)
	        \geq w(E(O^{out})),\]
	    where we use Claim~\ref{claim:ineg-k} for the first inequality (for $i = r$, \emph{i.e.}, $|\tilde{S}| \geq |O^{out}|$), Claim~\ref{claim:replace-greater} in the second inequality, and then we use that the sum of the weighted degrees is always greater than the actual weight of covered hyper-edges (because the hype-edges may be counted multiple times in the sum of the weighted degrees) for the last one.

	    As a result, we finally get:
	    \begin{multline*}
	        \mathbb{E}[w(E(S))] \geq w(E(O^{in})) + (1 - (\mu - 1)/\rho) \cdot w(E(O^{out})) - (\mu - 1)/\rho \cdot w(E(O^{in}))\\
	        \geq (1 - (\mu - 1)/\rho) \cdot w(E(O)).
	    \end{multline*}
	    Therefore by averaging principle, there exists $S^* \subseteq V'$ such that $S^* \in \mathcal{I}$ and \[w(E(S^*)) \geq (1 - (\mu - 1)/\rho) \cdot w(E(O))\]
     \end{proof}
	
	Now we give another interpretation of Algorithm~\ref{algo:kernel-construction} in terms of matroid union, which allows us to give a simpler description of the kernel $V'$. First recall the definition of matroid union:
	
	\definitionmatroidunion*

    In Algorithm~\ref{algo:simple-kernel-construction} we provide a simpler description of Algorithm~\ref{algo:kernel-construction} (the equivalence of the two algorithms is proved in Proposition~\ref{prop:equiv-algo-kernel}).
    
    \begin{algorithm}
	\caption{Algorithm for building a maximum coverage approximate kernel}\label{algo:simple-kernel-construction}
	\begin{algorithmic}[1]
	\State $V = \{v_1,\dots, v_n\}$ where $\deg_w(v_1) \geq \dots \geq \deg_w(v_n)$
	\State $V' \gets \emptyset$
	\For{$i = 1, \dots, n$} \Comment the vertices are processed in non-increasing order of weighted degree
	    \If{$V' \cup \{v_i\} \in \mathcal{I}_{\rho}$}
	        \State $V' \gets V' \cup \{v_i\}$
	    \EndIf
	\EndFor
	\State\Return $V'$
	\end{algorithmic}
	\end{algorithm}
    
    \begin{proposition} \label{prop:equiv-algo-kernel}
        Algorithm~\ref{algo:kernel-construction} and Algorithm~\ref{algo:simple-kernel-construction} build the same kernel $V'$. Moreover, $V'$ is a maximum weight independent set in $\rho \mathcal{M}$ with respect to the weighted degrees.
    \end{proposition}
    
    \begin{proof}
        To prove the first part of the proposition, it suffices to prove that the condition in Line~\ref{line:kernel-dismiss} of Algorithm~\ref{algo:kernel-construction} is equivalent to check whether $V' \cup C \cup \{v_i\}$ is or is not in $\mathcal{I}_{\rho}$ (if so, then $V' \cup C$ in Algorithm~\ref{algo:kernel-construction} plays the role of $V'$ in Algorithm~\ref{algo:simple-kernel-construction}).
        
        Using the decomposition in equation~(\ref{eq:kernel-decomp}), we know that when $v_i$ is processed and we check the condition in Line~\ref{line:kernel-dismiss} of Algorithm~\ref{algo:kernel-construction}, we have $V' = X_1 \cup \dots \cup X_j$ for some $j \in \{0, \dots, r\}$. We know that each $X_l$ is a $\rho$-DBS in $\mathcal{M}/(\bigcup_{l'=1}^{l-1}X_{l'})$, hence each one can be partitioned into $\rho$ independent sets $B_{l,1}, \dots, B_{l,\rho}$ in $\mathcal{M}/(\bigcup_{l'=1}^{l-1}X_{l'})$, all of size $\rnk_{\mathcal{M}/(\bigcup_{l'=1}^{l-1}X_{l'})}(X_l)$. Therefore,
        \[V' = \underbrace{(B_{1,1} \cup \dots \cup B_{j,1})}_{B_1} \cup \dots \cup \underbrace{(B_{1, \rho} \cup \dots \cup B_{j, \rho})}_{B_{\rho}}\]
        can be partitioned into $\rho$ independent sets in $\mathcal{M}$, each of size $\rnk_{\mathcal{M}}(V')$. Now, regarding condition in Line~\ref{line:kernel-dismiss} of Algorithm~\ref{algo:kernel-construction}:
        \begin{itemize}
            \item If $v_i \in \spn_{\mathcal{M}}(V')$, as $V'$ is a set containing $\rho \cdot \rnk_{\mathcal{M}}(X)$ elements, the set $V' \cup \{v_i\}$ will contain $\rho \cdot \rnk_{\mathcal{M}}(X) + 1$ elements while still being of rank equal to $\rnk_{\mathcal{M}}(V')$: hence it is not possible to partition $V' \cup \{v_i\}$ into $\rho$ independent sets (as each one can contain up to $\rnk_{\mathcal{M}}(V')$ elements), so $V' \cup C \cup \{v_i\} \not\in \mathcal{I}_{\rho}$.
            \item Otherwise, by Claim~\ref{claim:line-ineg-kernel}, the densest subset in $C$ with respect to the matroid $\mathcal{M}/V'$ is of density strictly below $\rho$, hence, by Proposition~\ref{prop:add-one-densest}, the densest subset in $C \cup \{v_i\}$ is of density at most $\rho$. We can therefore use Theorem~\ref{thm:edm-partition} to partition $C \cup \{v_i\}$ into $\rho$ independent sets $C_1, \dots, C_{\rho}$ in $\mathcal{M}/V'$, hence $V' \cup C \cup \{v_i\} = (B_1 \cup C_1) \cup \dots \cup (B_{\rho} \cup C_{\rho})$ can be partitioned into $\rho$ independent subsets, \emph{i.e.}, $V' \cup C \cup \{v_i\} \in \mathcal{I}_{\rho}$.
        \end{itemize}
        
        Therefore Algorithms~\ref{algo:kernel-construction} and~\ref{algo:simple-kernel-construction} build the very same approximate kernel $V'$. As Algorithm~\ref{algo:simple-kernel-construction} is simply the greedy algorithm to build a maximum weight independent set in the matroid $\rho \mathcal{M}$ with respect to the weighted degrees (see~\cite{Edmonds1971}), $V'$ is a maximum weight independent set in $\rho \mathcal{M}$.
    \end{proof}
    
    \theoremunion*
    
    \begin{proof}
        In fact, from the characterization of maximum weight independent sets in~\cite{Edmonds1971}, any maximum weight independent set in $\rho \mathcal{M}$ can be obtained by choosing the right processing order in Algorithm~\ref{algo:simple-kernel-construction} (\emph{i.e.}, for elements having the same weighted degrees, putting first the ones we want to pick in our kernel). Hence that kernel would have been built following the procedure of Algorithm~\ref{algo:kernel-construction}, and therefore the result of Lemma~\ref{lem:ratio-kernel} applies.
    \end{proof}
    
\bibliography{library}

\begin{thebibliography}{FKLM20}

\bibitem[AS04]{AgeevS04}
Alexander~A. Ageev and Maxim Sviridenko.
\newblock Pipage rounding: {A} new method of constructing algorithms with
  proven performance guarantee.
\newblock {\em J. Comb. Optim.}, 8(3):307--328, 2004.

\bibitem[BPS16]{BonnetPS16}
{\'{E}}douard Bonnet, Vangelis~Th. Paschos, and Florian Sikora.
\newblock Parameterized exact and approximation algorithms for maximum
  \emph{k}-set cover and related satisfiability problems.
\newblock {\em {RAIRO} Theor. Informatics Appl.}, 50(3):227--240, 2016.

\bibitem[CCPV11]{CalinescuCPV11}
Gruia C{\u{a}}linescu, Chandra Chekuri, Martin P{\'{a}}l, and Jan
  Vondr{\'{a}}k.
\newblock Maximizing a monotone submodular function subject to a matroid
  constraint.
\newblock {\em {SIAM} J. Comput.}, 40(6):1740--1766, 2011.

\bibitem[CVZ10]{ChekuriVZ10}
Chandra Chekuri, Jan Vondr{\'{a}}k, and Rico Zenklusen.
\newblock Dependent randomized rounding via exchange properties of
  combinatorial structures.
\newblock In {\em 51th Annual {IEEE} Symposium on Foundations of Computer
  Science, {FOCS} 2010, October 23-26, 2010, Las Vegas, Nevada, {USA}}, pages
  575--584. {IEEE} Computer Society, 2010.

\bibitem[Edm65]{edmonds1965minimumpartition}
Jack Edmonds.
\newblock Minimum partition of a matroid into independent subsets.
\newblock {\em J. Res. Nat. Bur. Standards Sect. B}, 69:67--72, 1965.

\bibitem[Edm71]{Edmonds1971}
Jack Edmonds.
\newblock Matroids and the greedy algorithm.
\newblock {\em Mathematical Programming}, 1(1):127--136, 1971.

\bibitem[Fei98]{Feige98}
Uriel Feige.
\newblock A threshold of $\ln n$ for approximating set cover.
\newblock {\em J. {ACM}}, 45(4):634--652, 1998.

\bibitem[FKLM20]{FeldmannSLM20}
Andreas~Emil Feldmann, {Karthik {C. S.}}, Euiwoong Lee, and Pasin Manurangsi.
\newblock A survey on approximation in parameterized complexity: Hardness and
  algorithms.
\newblock {\em Algorithms}, 13(6):146, 2020.

\bibitem[FW12]{FilmusW12-cover}
Yuval Filmus and Justin Ward.
\newblock The power of local search: Maximum coverage over a matroid.
\newblock In Christoph D{\"{u}}rr and Thomas Wilke, editors, {\em 29th
  International Symposium on Theoretical Aspects of Computer Science, {STACS}
  2012, February 29th - March 3rd, 2012, Paris, France}, volume~14 of {\em
  LIPIcs}, pages 601--612. Schloss Dagstuhl - Leibniz-Zentrum f{\"{u}}r
  Informatik, 2012.

\bibitem[FW14]{FilmusW14-submod}
Yuval Filmus and Justin Ward.
\newblock Monotone submodular maximization over a matroid via non-oblivious
  local search.
\newblock {\em {SIAM} J. Comput.}, 43(2):514--542, 2014.

\bibitem[GNW05]{GuoNW05}
Jiong Guo, Rolf Niedermeier, and Sebastian Wernicke.
\newblock Parameterized complexity of generalized vertex cover problems.
\newblock In Frank K. H.~A. Dehne, Alejandro L{\'{o}}pez{-}Ortiz, and
  J{\"{o}}rg{-}R{\"{u}}diger Sack, editors, {\em Algorithms and Data
  Structures, 9th International Workshop, {WADS} 2005, Waterloo, Canada, August
  15-17, 2005, Proceedings}, volume 3608 of {\em Lecture Notes in Computer
  Science}, pages 36--48. Springer, 2005.

\bibitem[HP98]{Hochbaum1998}
Dorit~S Hochbaum and Anu Pathria.
\newblock Analysis of the greedy approach in problems of maximum k-coverage.
\newblock {\em Naval Research Logistics (NRL)}, 45(6):615--627, 1998.

\bibitem[HS22]{HuangS22-swat}
Chien{-}Chung Huang and Fran{\c{c}}ois Sellier.
\newblock Matroid-constrained maximum vertex cover: Approximate kernels and
  streaming algorithms.
\newblock In Artur Czumaj and Qin Xin, editors, {\em 18th Scandinavian
  Symposium and Workshops on Algorithm Theory, {SWAT} 2022, June 27-29, 2022,
  T{\'{o}}rshavn, Faroe Islands}, volume 227 of {\em LIPIcs}, pages
  27:1--27:15. Schloss Dagstuhl - Leibniz-Zentrum f{\"{u}}r Informatik, 2022.

\bibitem[HS23]{HuangS23}
Chien{-}Chung Huang and Fran{\c{c}}ois Sellier.
\newblock Matroid-constrained vertex cover.
\newblock {\em Theor. Comput. Sci.}, 965:113977, 2023.

\bibitem[HW20]{HuangW2020}
Chien{-}Chung Huang and Justin Ward.
\newblock {FPT}-algorithms for the $l$-matchoid problem with linear and
  submodular objectives.
\newblock {\em CoRR}, abs/2011.06268, 2020.

\bibitem[Kam22]{Kamiyama22-robust}
Naoyuki Kamiyama.
\newblock A note on robust subsets of transversal matroids.
\newblock {\em CoRR}, abs/2210.09534, 2022.

\bibitem[LPRS17]{LokshtanovPRS17}
Daniel Lokshtanov, Fahad Panolan, M.~S. Ramanujan, and Saket Saurabh.
\newblock Lossy kernelization.
\newblock In Hamed Hatami, Pierre McKenzie, and Valerie King, editors, {\em
  Proceedings of the 49th Annual {ACM} {SIGACT} Symposium on Theory of
  Computing, {STOC} 2017, Montreal, QC, Canada, June 19-23, 2017}, pages
  224--237. {ACM}, 2017.

\bibitem[Man19]{Manurangsi19-sosa}
Pasin Manurangsi.
\newblock A note on max k-vertex cover: Faster fpt-as, smaller approximate
  kernel and improved approximation.
\newblock In Jeremy~T. Fineman and Michael Mitzenmacher, editors, {\em 2nd
  Symposium on Simplicity in Algorithms, {SOSA} 2019, January 8-9, 2019, San
  Diego, CA, {USA}}, volume~69 of {\em OASIcs}, pages 15:1--15:21. Schloss
  Dagstuhl - Leibniz-Zentrum f{\"{u}}r Informatik, 2019.

\bibitem[Man20]{Manurangsi20}
Pasin Manurangsi.
\newblock Tight running time lower bounds for strong inapproximability of
  maximum $k$-coverage, unique set cover and related problems (via $t$-wise
  agreement testing theorem).
\newblock In Shuchi Chawla, editor, {\em Proceedings of the 2020 {ACM-SIAM}
  Symposium on Discrete Algorithms, {SODA} 2020, Salt Lake City, UT, USA,
  January 5-8, 2020}, pages 62--81. {SIAM}, 2020.

\bibitem[Mar08]{Marx08}
D{\'{a}}niel Marx.
\newblock Parameterized complexity and approximation algorithms.
\newblock {\em Comput. J.}, 51(1):60--78, 2008.

\bibitem[MTV21]{McGregor21-cover}
Andrew McGregor, David Tench, and Hoa~T. Vu.
\newblock Maximum coverage in the data stream model: Parameterized and
  generalized.
\newblock In Ke~Yi and Zhewei Wei, editors, {\em 24th International Conference
  on Database Theory, {ICDT} 2021, March 23-26, 2021, Nicosia, Cyprus}, volume
  186 of {\em LIPIcs}, pages 12:1--12:20. Schloss Dagstuhl - Leibniz-Zentrum
  f{\"{u}}r Informatik, 2021.

\bibitem[NWF78]{nemhauser1978analysis}
George~L Nemhauser, Laurence~A Wolsey, and Marshall~L Fisher.
\newblock An analysis of approximations for maximizing submodular set
  functions—i.
\newblock {\em Mathematical programming}, 14(1):265--294, 1978.

\bibitem[Sch03]{Sch2003}
Alexander Schrijver.
\newblock {\em Combinatorial optimization: polyhedra and efficiency},
  volume~24.
\newblock Springer Science \& Business Media, 2003.

\bibitem[SF15]{SkowronF15}
Piotr~Krzysztof Skowron and Piotr Faliszewski.
\newblock Fully proportional representation with approval ballots:
  Approximating the maxcover problem with bounded frequencies in {FPT} time.
\newblock In Blai Bonet and Sven Koenig, editors, {\em Proceedings of the
  Twenty-Ninth {AAAI} Conference on Artificial Intelligence, January 25-30,
  2015, Austin, Texas, {USA}}, pages 2124--2130. {AAAI} Press, 2015.

\bibitem[Sko17]{Skowron17}
Piotr Skowron.
\newblock {FPT} approximation schemes for maximizing submodular functions.
\newblock {\em Inf. Comput.}, 257:65--78, 2017.

\bibitem[Sri01]{Srinivasan01}
Aravind Srinivasan.
\newblock Distributions on level-sets with applications to approximation
  algorithms.
\newblock In {\em 42nd Annual Symposium on Foundations of Computer Science,
  {FOCS} 2001, 14-17 October 2001, Las Vegas, Nevada, {USA}}, pages 588--597.
  {IEEE} Computer Society, 2001.

\end{thebibliography}

\appendix

\section{Sequential Random Sampling in DBSes}
    \label{sec:random-seq}

    \paragraph*{Sampling Technique} Here we give a random sampling procedure providing weaker negative correlation guarantees than Proposition~\ref{prop:random-sampling} (as we guarantee only a pairwise negative correlation up to a factor $2$), but that takes advantage of the particular structure of DBSes to pick elements \emph{sequentially} in the Density-Balanced Subset. This approach may be of interest as it is very different from the randomized rounding developed in~\cite{ChekuriVZ10}.

    \begin{proposition} \label{prop:random-seq}
        Let $\mathcal{M} = (V, \mathcal{I})$ be a matroid, $V' \subseteq V$ be a $\rho$-DBS for some positive integer $\rho$, and $k = \rnk_{\mathcal{M}}(V')$. The exists a sequential procedure to sample randomly from $V'$ an independent set of $k$ elements $S = \{s_1, \dots, s_k\} \in \mathcal{I}$ such that:
        \begin{enumerate}[(i)]
            \item for all $v \in V'$, $\mathbb{P}[v \in S] = 1/\rho$;
            \item for all $v, v' \in V'$ such that $v \neq v'$, $\mathbb{P}[v \in S \wedge v' \in S] \leq 2 \cdot \mathbb{P}[v \in S] \cdot \mathbb{P}[v' \in S] = 2/\rho^2$.
        \end{enumerate}
    \end{proposition}

    \begin{proof}
        We proceed by induction on $k = \rnk_{\mathcal{M}}(V')$, hence picking elements one by one into $S$.
        
        Consider the case $k = 1$. The set $V'$ is made of $\rho$ element, and none of these elements is a loop (if some $v \in V'$ were a loop, then $\{v\}$ would be a set in infinite density, contradicting that $V'$ was a $\rho$-DBS). Therefore by putting into $S$ one element $s_1$ chosen uniformly at random in $V'$ we get the desired result.
        
        Now let $k \geq 2$ and suppose that the result is true for any $\rho$-DBS of rank $k - 1$, in any matroid. Let $V'$ be a $\rho$-DBS of rank $k$. In the following, we will use an arbitrary indexing $v_1,\dots,v_{k\rho}$ of the elements of $V' = \{v_1,\dots,v_{k\rho}\}$.
        
        By Proposition~\ref{prop:dbs-contraction}, for each $v_i \in V'$, there exists a set $V'_i \subseteq V'\backslash\{v_i\}$ such that $V'_i$ is a $\rho$-DBS in $\mathcal{M}/\{v_i\}$ and $\rnk_{\mathcal{M}/\{v_i\}}(V'_i) = k - 1$. Moreover, by the induction hypothesis, we know that it is possible to sample uniformly at random from $V'_i$ an independent set $S_i = \{s_1,\dots,s_{k-1}\}$ in $\mathcal{M}/\{v_i\}$ satisfying properties~(i) and~(ii).
        
        Here is how we will randomly sample an independent set $S$ of $k$ elements from $V'$. We will choose a probability distribution $\{p_1, \dots, p_{k\rho}\}$ on $V'$. Then we will build $S$ by first sampling the element $s_k$ from $V'$ according to this distribution (\emph{i.e.}, sampling $v_i$ with probability $p_i$) and then sample the remaining elements $S_i = \{s_1, \dots, s_{k-1}\}$ recursively in the contracted $\rho$-DBS $V'_i$, so that $S_i$ is independent in $\mathcal{M}/\{v_i\}$.
        
        Therefore, we know that the probability that $v_i \in S$ is:
        \begin{align*}
            \mathbb{P}[v_i \in S] &= \mathbb{P}[s_k = v_i] + \sum_{j \neq i}\mathbb{P}[s_k = v_j] \cdot \mathbb{P}[v_i \in S \,|\, s_k = v_j]\\
            &= p_i + \sum_{j : v_i \in V'_{j}}p_j \cdot 1/\rho\\
            &= p_i + \frac{1}{\rho} \cdot \left(1 - \sum_{j : v_i \not\in V'_{j}} p_j\right),
        \end{align*}
        where to move from the first to the second line we used the induction hypothesis (property~(i)).
        
        Now we want to choose the $p_i$s so that for all $i$, we have $\mathbb{P}[v_i \in S] = 1/\rho$, hence this leads to the following system of linear equations:
        \begin{equation} \label{eq:syst-pi}
            \forall\,1 \leq i \leq k \cdot \rho, \quad p_i = \frac{1}{\rho} \sum_{j : v_i \not\in V'_{j}} p_j
        \end{equation}
        To prove that there exists probability distribution $\{p_i\}_{1 \leq i \leq k\rho}$ satisfying the constraints of~(\ref{eq:syst-pi}) we introduce the $k\rho \times k\rho$ matrix $T$ defined as follows:
        \[
	        T_{i,j} = \left\{
            \begin{array}{ll}
                1/\rho & \text{if } v_j \not\in V'_{i} \\
                0 & \text{otherwise}
            \end{array}
        \right.
	    \]
	    For all $1 \leq i \leq k \cdot \rho$, by construction of the $V'_{i}$s, exactly $\rho$ elements $v_j$ are not in $V'_{i}$. Hence $T$ is a real matrix with non-negative entries and each of its rows sums to $1$. Therefore, $T$ is a right stochastic matrix.  As a result, there exists a stationary distribution $\mu = (\mu_1, \dots, \mu_{k \rho})$ such that $\mu = \mu T$. By setting $p_i = \mu_i$, we obtain the desired probability distribution, as being a solution of~(\ref{eq:syst-pi}) is equivalent to being an eigenvector of the matrix $T$ with eigenvalue $1$.
	    
	    Hence the choice of the $p_i$s guarantees that property~(i) is satisfied. Now we want to prove that our sampling also satisfies property~(ii), $\emph{i.e.}$, given any $i \neq j$, we have $\mathbb{P}[v_i \in S \wedge v_j \in S] \leq 2/\rho^2$. And in fact, we have
	    \begin{align*}
	        \mathbb{P}[v_i \in S \wedge v_j \in S] &= \sum_{l = 1}^{k\rho} \mathbb{P}[s_k = v_l] \cdot \mathbb{P}[v_i \in S \wedge v_j \in S \,|\, s_k = v_l]\\
	        &= p_i \cdot \mathbb{P}[v_j \in S \,|\, s_k = v_i] + p_j \cdot \mathbb{P}[v_i \in S \,|\, s_k = v_j]\\
	        &\quad\quad+ \sum_{l \neq i,j} p_l \cdot \mathbb{P}[v_i \in S \wedge v_j \in S \,|\, s_k = v_l]\\
	        &\leq p_i \cdot 1/\rho + p_j \cdot 1/\rho + \sum_{l : v_i, v_j \in V'_l} p_l \cdot 2/\rho^2\\
	        &= \frac{1}{\rho^2} \cdot \left(\sum_{l : v_i \not\in V'_l} p_l + \sum_{l : v_j \not\in V'_l} p_l + 2 \sum_{l : v_i, v_j \in V'_l} p_l\right) &\text{by~(\ref{eq:syst-pi})}\\
	        &\leq \frac{1}{\rho^2} \cdot \left(\sum_{l : v_i \not\in V'_l} p_l + \sum_{l : v_j \not\in V'_l} p_l + \sum_{l : v_i \in V'_l} p_l + \sum_{l : v_j \in V'_l} p_l\right)\\
	        &= \frac{2}{\rho^2},
	    \end{align*}
	    where in the first inequality we use the induction hypothesis and that if $v_i$ (resp $v_j$) is not in $V'_{j}$ (resp $V'_{i}$) then $\mathbb{P}[v_i \in S \,|\, s_k = v_j] = 0$ (resp $\mathbb{P}[v_j \in S \,|\, s_k = v_i] = 0$). This concludes the induction step.
    \end{proof}

    \paragraph*{Example for a Laminar Matroid} We provide here an example of random sampling procedure for the $\rho$-DBS of Figure~\ref{fig:laminar-matroid}, with $\rho = 4$. First, a subset $V'_i$ is associated to each $v_i$, as described in the proof of Proposition~\ref{prop:random-seq}: $V'_i$ has to be a $\rho$-DBS of rank $1$ in $\mathcal{M}/\{v_i\}$; such sets $V'_i$ are provided in the second column of Table~\ref{tab:example}. Then based on the choice of these sets, we can choose a probability distribution $\{p_i\}_{1 \leq i \leq 2\rho}$ satisfying the system~(\ref{eq:syst-pi}), \emph{i.e.} being a fixed point of the matrix $T$~(\ref{eq:matrix-example}); such probability distribution is given in the third column of Table~\ref{tab:example}.

    \begin{table}[h]
    \centering
    \begin{tabular}{c|c|c} 
         $v_i$ & $V'_i$ & $p_i$\\
         \hline
         $v_1$ & $\{v_3, v_4, v_5, v_6\}$ & $1/12$\\
         $v_2$ & $\{v_3, v_4, v_5, v_6\}$ & $1/4$\\
         $v_3$ & $\{v_1, v_2, v_6, v_7\}$ & $0$\\
         $v_4$ & $\{v_1, v_2, v_6, v_7\}$ & $0$\\
         $v_5$ & $\{v_1, v_2, v_6, v_7\}$ & $0$\\
         $v_6$ & $\{v_1, v_3, v_4, v_5\}$ & $1/6$\\
         $v_7$ & $\{v_1, v_3, v_4, v_5\}$ & $1/4$\\
         $v_8$ & $\{v_1, v_3, v_4, v_5\}$ & $1/4$\\
    \end{tabular}
    \caption{\label{tab:example}Parameters for the random sampling procedure derived from the proof of Proposition~\ref{prop:random-seq} for the $\rho$-DBS depicted in Figure~\ref{fig:laminar-matroid} ($\rho = 4$).}
    \end{table}

    \begin{equation}\label{eq:matrix-example}
        T = 
        \begin{pmatrix}
            1/4 & 1/4 & 0 & 0 & 0 & 0 & 1/4 & 1/4\\
            1/4 & 1/4 & 0 & 0 & 0 & 0 & 1/4 & 1/4\\
            0 & 0 &  1/4 & 1/4 & 1/4 & 0 & 0 &  1/4\\
            0 & 0 &  1/4 & 1/4 & 1/4 & 0 & 0 &  1/4\\
            0 & 0 &  1/4 & 1/4 & 1/4 & 0 & 0 &  1/4\\
            0 & 1/4 & 0 & 0 & 0 & 1/4 & 1/4 & 1/4\\
            0 & 1/4 & 0 & 0 & 0 & 1/4 & 1/4 & 1/4\\
            0 & 1/4 & 0 & 0 & 0 & 1/4 & 1/4 & 1/4
        \end{pmatrix}
    \end{equation}

    At the first step, an element $v_i$ is sampled with probability $p_i$ and then the second element is chosen uniformly at random among the elements of the associated set $V'_i$. We can observe here that the elements $v_3$, $v_4$, and $v_5$ are not sampled in the first step, but that their probability to be sampled during the procedure is still $1/\rho$ because they will always appear in $V'_i$ for the second round. Similarly, $p_8$ is set $1/4$ because it will never appear in the second round. One can easily check that this procedure guarantees that each element has the same probability of being sampled and that the elements are (almost) negatively correlated. 
    
\end{document}